\documentclass[a4paper,UKenglish,runningheads,11pt]{llncs}

\usepackage{tabularx,booktabs,multirow,delarray,array}
\usepackage{graphicx,amssymb,amsmath}
\usepackage{enumerate}
\usepackage[ruled,vlined,linesnumbered]{algorithm2e}
\usepackage{fullpage}
\usepackage{latexsym}
\usepackage{enumerate}
\usepackage{lineno}
\usepackage{wrapfig}
\usepackage{hyperref}
\usepackage{bibspacing}
\newenvironment{proof}{\par\noindent{\bf Proof:}}{\mbox{}\hfill$\qed$\\}

\newcommand{\ignore}[1]{ }

\newcounter{rem}
\def\qed{\hbox{\rlap{$\sqcap$}$\sqcup$}}

\begin{document}

\title{Vertex Fault-Tolerant Spanners for Weighted Points in Polygonal Domains}
\titlerunning{Fault-tolerant spanners for weighted points in polygonal domains}

\author{
R. Inkulu\inst{1}
\and
Apurv Singh\inst{1}
}

\authorrunning{R. Inkulu, A. Singh}

\institute{
Department of Computer Science and Engineering\\
IIT Guwahati, India\\
\email{\{rinkulu,apursingh\}@iitg.ac.in}
}

\maketitle

\pagenumbering{arabic}
\setcounter{page}{1}

\begin{abstract}
Given a set $S$ of $n$ points, a weight function $w$ to associate a non-negative weight to each point in $S$, a positive integer $k \ge 1$, and a real number $\epsilon > 0$, we devise the following algorithms to compute a $k$-vertex fault-tolerant spanner network $G(S, E)$ for the metric space induced by the weighted points in $S$:
(1) When the points in $S$ are located in a simple polygon, we present an algorithm to compute $G$ with multiplicative stretch $\sqrt{10}+\epsilon$, and the number of edges in $G$ (size of $G$) is $O(k n (\lg{n})^2)$.
(2) When the points in $S$ are located in the free space of a polygonal domain $\cal P$ with $h$ number of obstacles, we present an algorithm to compute $G$ with multiplicative stretch $6+\epsilon$ and size $O(\sqrt{h} k n(\lg{n})^2)$.
(3) When the points in $S$ are located on a polyhedral terrain, we devise an algorithm to compute $G$ with multiplicative stretch $6+\epsilon$ and size $O(k n (\lg{n})^2)$.
\end{abstract}


\section{Introduction}
\label{sect:intro}

In designing geometric networks on a given set of points in a metric space, it is desirable for the network to have short paths between any pair of nodes while being sparse with respect to the number of edges.
For a set $S$ of $n$ points in a metric space $\cal M$, a network on $S$ is an undirected graph $G$ with vertex set $S$ and an edge set $E$, where every edge of $G$ is associated with a weight. 
The distance in $G$ between any two vertices $p$ and $q$ of $G$, denoted by $d_G(p, q)$, is the length of a shortest (that is, a minimum length) path between $p$ and $q$ in $G$.
For a real number $t \ge 1$, the graph $G$ is called a {\it $t$-spanner} of points in $S$ if for every two points $p, q \in S$, $d_G(p, q)$ is at most $t$ times the distance between $p$ and $q$ in $\cal M$.
The smallest $t$ for which $G$ is a $t$-spanner of points in $S$ is called the {\it stretch factor} of $G$, and the number of edges of $G$ is called its size.
Given a set $S$ of points in the plane, each associated with a non-negative weight, and a positive integer $k$, in this paper, we study computing an edge-weighted geometric graph $G$ that is a {\it vertex fault-tolerant $t$-spanner} for the metric space induced by the weighted points in $S$;
that is, for any set $S' \subseteq S$ of cardinality at most $k$, the graph $G \setminus S'$ is a $t$-spanner for the metric space induced by the weighted points in $S-S'$.

\subsection*{Previous Work}

Peleg and Sch\"{a}ffer~\cite{journals/jgt/PelegSchaffer89} and Chew~\cite{journals/jcss/Chew89} introduced spanner networks.
Alth\"{o}fer et~al.~\cite{journals/dcg/AlthoferDDJS93} studied sparse spanners on edge-weighted graphs with edge weights obeying the triangle-inequality.
The text by Narasimhan and Smid \cite{books/compgeom/narsmid2007}, and the handbook chapters by Eppstein \cite{hb/cg/Epp99} and Gudmundsson and Knauer~\cite{hb/apprxheu/GudKnau07} detail various results on Euclidean spanners, including a $(1+\epsilon)$-spanner for the set $S$ of $n$ points in $\mathbb{R}^d$ that has $O(\frac{n}{\epsilon^{d-1}})$ edges, for any $\epsilon > 0$.

Apart from the small number of edges, spanners with additional properties such as small weight, bounded degree, small diameter, planar network, etc., were also considered.
The significant results in optimizing these parameters in geometric spanner network design include 
spanners of low degree \cite{journals/cgta/ABCGHSV08,conf/cccg/CarmiChai10,journals/cgta/BCCCKL13}, 
spanners of low weight \cite{journals/algorithmica/BCFMS10,journals/ijcga/DasNara97,journals/siamjc/GudmudLevcoNar02}, 
spanners of low diameter \cite{conf/focs/AryaMS94,journals/cgta/AryaMountSmid99}, 
planar spanners \cite{conf/esa/ArikatiCCDSZ96,journals/jcss/Chew89,conf/optalgo/DasJoseph89,conf/wads/KeilGutwin89}, 
spanners of low chromatic number \cite{journals/cgta/BCCMSZ09}, 
fault-tolerant spanners \cite{journals/dcg/ABFG09,journals/dcg/CzumajZ04,journals/algorithmica/LevcopoulosNS02,conf/wads/Lukovszki99,journals/corr/KapLi13,conf/stoc/Solomon14,conf/caldam/Inkulu19b}, 
low power spanners \cite{journals/wireless/Karim11,conf/infocom/SegalShpu10,journals/jco/WangLi06}, 
kinetic spanners \cite{journals/dcg/AbamBerg11,journals/cgta/ABG10}, 
angle-constrained spanners \cite{journals/jocg/CarmiSmid12}, 
and combinations of these \cite{conf/stoc/AryaDMSS95,journals/algorithmica/AryaS97,journals/algorithmica/BFRV18,journals/algorithmica/BoseGudSmid05,journals/ijcga/BoseSmidXu09,journals/corr/CarmiChait10}.
For the case of spanners in a metric space with bounded doubling metric, a few results are given in \cite{conf/stoc/Talwar04}.

As observed in Abam et~al.,~\cite{journals/algorithmica/AbamBFGS11}, the cost of traversing a path in a network is not only determined by the lengths of edges along the path, but also by the delays occurring at the vertices on the path.
The result in \cite{journals/algorithmica/AbamBFGS11} models these delays by associating non-negative weights to points. 
Let $S$ be a set of $n$ points in $\mathbb{R}^d$. 
For every $p \in S$, let $w(p)$ be the non-negative weight associated to $p$. 
The following weighted distance function $d_w$ on $S$ defining the metric space $(S, d_w)$ is considered by Abam et~al. in \cite{journals/algorithmica/AbamBFGS11}, and by Bhattacharjee and Inkulu in \cite{conf/caldam/Inkulu19b}:
for any $p, q \in S$, $d_w(p, q)$ is equal to $w(p) + |pq| + w(q)$ if $p \ne q$; otherwise, $d_w(p, q)$ is equal to $0$.

Recently, Abam et~al.~\cite{conf/soda/AbamBS17} showed that there exists a $(2 + \epsilon)$-spanner with a linear number of edges for the metric space $(S, d_w)$ that has bounded doubling dimension.
And, \cite{journals/algorithmica/AbamBFGS11} gives a lower bound on the stretch factor, showing that $(2+\epsilon)$ stretch is nearly optimal.
Bose et~al. \cite{conf/swat/BoseCC08} studied the problem of computing a spanner for a set of weighted points in ${\mathbb R}^2$, while defining the distance $d_w$ between any two distinct points $p, q \in S$ as $d(p,q) - w(p) - w(q)$. 
Under the assumption the distance between any pair of points is non-negative, they showed the existence of a $(1 + \epsilon)$-spanner with $O(\frac{n}{\epsilon})$ edges. 

A set of $h \ge 0$ disjoint simple polygonal holes (obstacles) contained in a simple polygon $P$ is the {\it polygonal domain} ${\cal P}$.
Note that when $h = 0$, the polygonal domain $\cal P$ is essentially a simple polygon.
The free space $\cal D$ of a polygonal domain ${\cal P}$ is defined as the closure of $P$ excluding the union of the interior of polygons contained in $P$.
Note that the free space of a simple polygon is its closure.
A shortest path between any two points $p$ and $q$ is a path in $\cal D$ whose length (in Euclidean metric) is less than or equal to the length of any path between $p$ and $q$ located in $\cal D$.
The distance along any shortest path between $p$ and $q$ is denoted by $d_\pi(p, q)$.
If the line segment joining $p$ and $q$ is in $\cal D$, then the Euclidean distance between $p$ and $q$ is denoted by $d(p, q)$, i.e., in this case, $d_\pi(p, q)$ is equal to $d(p, q)$.

Given a set $S$ of $n$ points in the free space $\cal D$ of $\cal P$, computing a geodesic spanner of $S$ is considered in Abam et~al.~\cite{conf/compgeom/AbamAHA15}.
The result in \cite{conf/compgeom/AbamAHA15} showed that for the metric space $(S, d_\pi)$, for any constant $\epsilon > 0$, there exists a $(5+\epsilon)$-spanner of size $O(\sqrt{h}n (\lg n)^{2})$. 
Further, when the input points are located in a simple polygon, for any constant $\epsilon > 0$, \cite{conf/compgeom/AbamAHA15} devised an algorithm to compute a $(\sqrt{10}+\epsilon)$-spanner with $O(n (\lg n)^{2})$ edges.

A polyhedral terrain $\cal T$ is the graph of a piecewise linear function $f: D \rightarrow {\mathbb R}^3$, where $D$ is a convex polygonal region in the plane.
Given a set $S$ of $n$ points on a polyhedral terrain $\mathcal{T}$, the geodesic distance between any two points $p, q \in S$ is the distance along any shortest path on the terrain between $p$ and $q$.
The spanner for points on a terrain with respect to geodesic distance on terrain is a geodesic spanner.
Unlike metric space induced by Euclidean distance among points in ${\mathbb R}^d$, the metric space induced by points on a terrain does not have a bounded doubling dimension.
Hence, geometric spanners for points on polyhedral terrains has unique characteristics and interesting to study.
The algorithm in \cite{conf/soda/AbamBS17} proved that for a set of unweighted points on any polyhedral terrain, for any constant $\epsilon > 0$, there exists a $(2 + \epsilon)$-geodesic spanner with $O(n \lg n)$ edges.

A graph $G(S, E)$ is a {\it $k$-vertex fault-tolerant $t$-spanner}, denoted by $(k, t)$-VFTS, for a set $S$ of $n$ points in $\mathbb{R}^d$ whenever for any subset $S'$ of $S$ with size at most $k$, the graph $G \setminus S'$ is a $t$-spanner for the points in $S \setminus S'$.
The algorithms given in Levcopoulos et~al., \cite{journals/algorithmica/LevcopoulosNS02}, Lukovszki \cite{conf/wads/Lukovszki99}, and Czumaj and Zhao \cite{journals/dcg/CzumajZ04} compute a $(k, t)$-VFTS for the set $S$ of points in $\mathbb{R}^d$.
These algorithms are also presented in \cite{books/compgeom/narsmid2007}.
Levcopoulos et~al.~\cite{journals/algorithmica/LevcopoulosNS02} devised an algorithm to compute a $(k, t)$-VFTS of size $O(\frac{n}{(t-1)^{(2d -1)(k+1)}})$ in $O(\frac{n \lg{n}}{(t-1)^{4d -1}} + \frac{n}{(t-1)^{(2d -1)(k+1)}})$ time,
and another algorithm to compute a $(k, t)$-VFTS with $O(k^{2}n)$ edges in $O(\frac{k n \lg n}{(t-1)^{d}})$ time. 
The result in \cite{conf/wads/Lukovszki99} gives an algorithm to compute a $(k, t)$-VFTS of size $O(\frac{k n}{(t-1)^{d-1}})$ in $O(\frac{1}{(t-1)^{d}}(n \lg^{d-1} n \lg k + k n \lg \lg n))$ time. 
The algorithm in \cite{journals/dcg/CzumajZ04} computes a $(k, t)$-VFTS having $O(\frac{k n}{(t-1)^{d-1}})$ edges in $O(\frac{1}{(t-1)^{d-1}}(k n \lg^{d} n + nk^{2} \lg k))$ time with total weight of edges upper bounded by $O(\frac{k^{2} \lg n}{(t-1)^{d}})$ multiplicative factor of the weight of a minimum spanning tree of the given set of points. 

\subsection*{Our results}

In \cite{conf/caldam/Inkulu19b}, Bhattacharjee and Inkulu devised the following algorithms: one for computing a $(k, 4+\epsilon, w)$-VFTSWP when the input points are in $\mathbb{R}^d$, and the other for computing a $(k, 4+\epsilon, w)$-VFTSWP when the given points are in a simple polygon.
Further, in \cite{conf/cocoon/BInkulu19}, Bhattacharjee and Inkulu extended these algorithms to compute a $(k, 4+\epsilon, w)$-VFTSWP when the points are in a polygonal domain and when the points are located on a terrain.
In this paper, we show the following results for computing a $(k, t, w)$-VFTSWP:

\begin{itemize}
\item[*]
Given a simple polygon $\cal P$, a set $S$ of $n$ points located in $\cal P$, a weight function $w$ to associate a non-negative weight to each point in $S$, a positive integer $k$, and a real number $0 < \epsilon \le 1$, we present an algorithm to compute a $(k, \sqrt{10}+\epsilon, w)$-VFTSWP that has size $O(k n(\lg{n})^2)$.  
(Refer to Theorem~\ref{thm:simppoly}.)
The stretch factor of the spanner is improved from the result in \cite{conf/caldam/Inkulu19b}, and the number of edges is an improvement over the result in \cite{conf/caldam/Inkulu19b} when $(\lg{n}) < \frac{1}{\epsilon^2}$.
Note that \cite{conf/caldam/Inkulu19b} devised an algorithm for computing a $(k, 4+\epsilon, w)$-VFTSWP with size $O(\frac{k n}{\epsilon^2} \lg{n})$.

\medskip

\item[*]
Given a polygon domain $\cal P$ with $h$ number of obstacles (holes), a set $S$ of $n$ points located in the free space $\cal D$ of $\cal P$, a weight function $w$ to associate a non-negative weight to each point in $S$, a positive integer $k$, and a real number $0 < \epsilon \le 1$, we present an algorithm to compute a $(k, 6+\epsilon, w)$-VFTSWP with size $O(\sqrt{h+1} k n(\lg{n})^2)$.  
(Refer to Theorem~\ref{thm:polydom}.)
Though the stretch factor of the VFTSWP given in \cite{conf/cocoon/BInkulu19} is $(4+\epsilon)$, its size is $O(\frac{\sqrt{h+1} k n}{{\epsilon^2}} (\lg{n}))$.

\medskip

\item[*]
Given a polyhedral terrain $\cal T$, a set $S$ of $n$ points located on $\cal T$, a weight function $w$ to associate a non-negative weight to each point in $S$, a positive integer $k$, and a real number $0 < \epsilon \le 1$, we present an algorithm to compute a $(k, 6+\epsilon, w)$-VFTSWP with size $O(n k (\lg{n})^2)$.
(Refer to Theorem~\ref{thm:terr2}.)
Analogous to the points in the free space of a polygonal domain, the stretch factor of the VFTSWP given in \cite{conf/cocoon/BInkulu19} is $(4+\epsilon)$, however its size is $O(\frac{n k}{{\epsilon^2}} (\lg{n}))$.
\end{itemize}

The approach in achieving these improvements is different from results in \cite{conf/caldam/Inkulu19b,conf/cocoon/BInkulu19}.
Instead of clustering (like in \cite{conf/caldam/Inkulu19b,conf/cocoon/BInkulu19}), following \cite{conf/compgeom/AbamAHA15}, our algorithm uses the $s$-semi-separated pair decomposition ($s$-SSPD) of points projected on line segments together with the divide-and-conquer applied to input points.
For a set $Q$ of $n$ points in $\mathbb{R}^d$, a {\it pair decomposition} \cite{conf/focs/Varada98,journals/cgta/AbamHarpeled12} of $Q$ is a set of pairs of subsets of $Q$, such that for every pair of points of $p, q \in Q$, there exists a pair $(A, B)$ in the decomposition such that $p \in A$ and $q \in B$.
Given a pair decomposition $\{\{A_1, B_1\}, \ldots, \{A_s, B_s\}\}$ of a point set, its weight is defined as $\sum_{i=1}^s (|A_i| + |B_i|)$.
For a set $Q$ of $n$ points in $\mathbb{R}^d$, a {\it $s$-semi-separated pair decomposition (s-SSPD)} of $Q$ is a pair decomposition of $Q$ such that for every pair $(A, B)$, the distance between $A$ and $B$ (i.e., the distance of their minimum enclosing disks) is greater than or equal to $s$ times the minimum of the radius of $A$ and the radius of $B$. 
(The radius of a point set $X$ is the radius of the smallest ball enclosing all the points in $X$.)
The SSPD has the advantage of low weight as compared to well-known well-separated pair decomposition.

\subsection*{Terminology}

Recall the Euclidean distance between two points $p$ and $q$ is denoted by $|pq|$ and the geodesic Euclidean distance between two points $p, q$ located in the free space of a polygonal domain is denoted by $d_\pi(p, q)$.
Here, $\pi$ denotes a shortest path between $p$ and $q$ located in the free space of the polygonal domain.
(Further, we would like to note the following terms are defined in previous subsection: polygonal domain, free space of a polygonal domain, and a geodesic shortest path between two points.)
We note that if the line segment joining $p$ and $q$ does not intersect any obstacle, $d_\pi(p, q)$ is equal to $|pq|$, otherwise $d_\pi(p, q)$ is the distance along a geodesic shortest path $\pi$ between $p$ and $q$.
The length of a shortest path between $p$ and $q$ in a graph $G$ is denoted by $d_G(p, q)$.
For a set $S'$ of vertices of $G$ and any two vertices $p, q$ of $G$ not belonging to $S'$, the distance along a shortest path between $p$ and $q$ in graph $G \setminus S'$ is denoted by $d_{G \backslash S'}(p, q)$.
As in \cite{journals/algorithmica/AbamBFGS11} and in \cite{conf/caldam/Inkulu19b}, the function $d_w$ is defined on a set $S$ of points as follows: for any $p, q \in S$, $d_w$ is equal to $w(p) + d_\pi(p, q) + w(q)$ if $p \ne q$; otherwise, $d_w(p, q)$ is equal to $0$. 

For any point $p$ and any line segment $\ell$, let $d$ be the geodesic (Euclidean) distance between $p$ and $\ell$.
Then any point $p_\ell \in \ell$ is called a {\it geodesic projection of $p$ on $\ell$} whenever the geodesic (Euclidean) distance between $p$ and $p_\ell$ is $d$.
We note that geodesic distance between two points is the Euclidean distance of a geodesic path located in the free space in case of a polygonal domain, and it is the distance along a geodesic path located on the polyhedral terrain. 
We denote a geodesic projection of a point $p$ on a line segment $\ell$ with $p_\ell$.

Recall that, for any set $S$ of points, any graph $G$ with vertex set $S$ and each of its edges associated with a non-negative weight is called a {\it $t$-spanner} of points in $S$ whenever $d_\pi(p, q) \le d_G(p, q) \le t \cdot d_\pi(p, q)$ for every two points $p, q \in S$ and a real number $t \ge 1$.
The smallest $t$ for which $G$ is a $t$-spanner of $S$ is called the {\it stretch factor} of $G$. 
The number of edges in $G$ is known as the size of $G$. 
A graph $G(S, E)$ is a {\it $k$-vertex fault-tolerant $t$-spanner}, denoted by $(k, t)$-VFTS, for a set $S$ of $n$ points whenever for any subset $S'$ of $S$ with size at most $k$, the graph $G \setminus S'$ is a $t$-spanner for the points in $S \setminus S'$.
For a real number $t > 1$ and a set $S$ of weighted points, a graph $G(S, E)$ is called a {\it $t$-spanner for weighted points in $S$} whenever $d_w(p, q) \le d_G(p, q) \le t \cdot d_w(p, q)$ for every two points $p$ and $q$ in $S$.
Given a set $S$ of points, a function $w$ to associate a non-negative weight to each point in $S$, an integer $k \geq 1$, and a real number $t > 0$, a geometric graph $G$ is called a {\it $(k, t, w)$-vertex fault-tolerant spanner for weighted points} in $S$, denoted by $(k, t, w)$-VFTSWP, whenever for any set $S' \subset S$ with cardinality at most $k$, the graph $G \setminus S'$ is a $t$-spanner for the weighted points in $S \setminus S'$.
Note that every edge $(p, q)$ in $G$ corresponds to a shortest geodesic path between two points $p, q \in S$.
In addition, the weight associated with any edge $(p, q)$ of $G$ is the distance $d_\pi(p,q)$ along a geodesic shortest path between $p$ and $q$.

\bigskip

Section~\ref{sect:simppoly} presents an algorithm to compute a $(k, \sqrt{10}+\epsilon, w)$-VFTSWP when the weighted input points are in a simple polygon. 
When the points are located in the free space of a polygonal domain, an algorithm for computing a $(k, 6+\epsilon, w)$-VFTSWP is detailed in Section~\ref{sect:polydom}.
Section~\ref{sect:terr} details an algorithm to compute a $(k, 6+\epsilon, w)$-VFTSWP when the points associated with non-negative weights are located on a polyhedral terrain.

\section{Vertex fault-tolerant spanner for weighted points in a simple polygon}
\label{sect:simppoly}

Given a simple polygon $\cal P$, a set $S$ of $n$ points located in $\cal P$, a weight function $w$ to associate a non-negative weight to each point in $S$, a positive integer $k$, and a real number $\epsilon > 0$, we devise an algorithm to compute a geodesic $(k, \sqrt{10}+\epsilon, w)$-VFTS for the set $S \backslash S'$ of weighted points.
That is, if for any set $S' \subset S$ with cardinality at most $k$, the graph $G \backslash S'$ is a $(\sqrt{10}+\epsilon)$-spanner for the set $S \backslash S'$ of weighted points.
To remind, a graph $G(S'', E'')$ is a $t$-spanner for the set $S''$ of weighted points located in a simple polygon $\cal P''$ whenever the distance between any two points $p, q \in S''$ in $G$ is upper bounded by $t \cdot (w(p) + d_\pi(p, q) + w(q))$.
Here, $d_\pi(p, q)$ is the geodesic distance between points $p$ and $q$ in the simple polygon $\cal P''$.

Using polygon-cutting theorem in \cite{conf/jcdcg/BoseCKKM98}, we partition $\cal P$ into two simple polygons $\cal P', \cal P''$ with a line segment $\ell \in \cal P$.
We project all the points in $S$ on to $\ell$; let $S_\ell$ be the set of points projected on to $\ell$.
We compute a $\frac{4}{\epsilon}$-SSPD $\cal S$ of points in $S_\ell$, and include edges into $G$ based on $\cal S$.
We recursively process points in simple polygons $\cal P'$ and $\cal P''$.
The details are in Algorithm~\ref{algo:vftswpsimppoly} listed below. 
Our algorithm extends the algorithm in \cite{conf/compgeom/AbamAHA15} to the case of input points associated with non-negative weights. 

\begin{algorithm}
	\caption{VFTSWPSimplePolygon($\cal P$, S).}
	\label{algo:vftswpsimppoly}

	\SetAlgoLined
	\SetKwInOut{Input}{Input}
	\SetKwInOut{Output}{Output}
	\Input{A simple polygon $\cal P$, a set $S$ on $n$ points located in $\cal P$, a weight function $w$ that associates a non-negative weight to each point in $S$, an integer $k \geq 1$, and a real number $0 < \epsilon \leq 1$.}
	\Output{A $(k, \sqrt{10}+\epsilon, w)$-VFTSWP $G$.}

	\BlankLine
    
	If $|S| \le 1$ then return. \\
    
	\BlankLine
    
	By using the polygon cutting theorem in \cite{conf/jcdcg/BoseCKKM98}, we partition $\cal P$ into two simple polygons ${\cal P}', {\cal P}''$ with a line segment $\ell$ joining two points on $\partial{\cal P}$ such that either of the sub-polygons contains at most two-thirds of the points in $S$. 
	Let $S'$ be the set of points in $\cal P'$, and let $S''$ be the set of points in $\cal P''$.
	(Without loss of generality, for any point $p \in S$, if $p$ is located on $\ell$, then we assume $p \in S'$ and $p \notin S''$.) \\

	\BlankLine
    
	For each point $p \in S$, compute a geodesic projection $p_\ell$ of $p$ on $\ell$.
	Let $S_\ell$ be the set of points resulting from the geodesic projection of each point in $S$ on $\ell$. \\

	\BlankLine
    
	Using the algorithm in \cite{journals/dcg/ABFG09}, compute a $\frac{4}{\epsilon}$-SSPD $\cal S$ for the points in $S_\ell$. \\

	\BlankLine

	IncludeEdgesUsingSSPD(${\cal S}, G$). \scriptsize{} (Refer to Algorithm~\ref{algo:addedgessspd}.) \normalsize{} \\

	\BlankLine

	VFTSWPSimplePolygon(${\cal P'}, S'$). \\

	\BlankLine

	VFTSWPSimplePolygon(${\cal P''}, S''$).

\end{algorithm}

\begin{algorithm}
	\caption{IncludeEdgesUsingSSPD(${\cal S}, G$)}
	\label{algo:addedgessspd}

	\SetAlgoLined
	\SetKwInOut{Input}{Input}
	\SetKwInOut{Output}{Output}
	\Input{An $s$-SSPD $\cal S$, and a graph $G$.}
	\Output{Based on $\cal S$, edges are added to $G$.}
	\BlankLine

	\ForEach{pair $(A,B)$ in $\cal S$}{
		\scriptsize{} (In the following, without loss of generality., we assume $radius(A) \leq radius(B)$.) \normalsize{} \\
			\BlankLine

		\uIf{$|A| < k+1$}{
			For every $p \in A$ and $q \in B$, add edge $(p,q)$ to $G$.
		}
		\Else{
			For every point $p$ in $A$, associate a weight $w(p) + d_\pi(p, p_\ell)$.  
			(Note that $d_\pi(p, p_\ell)$ is the geodesic distance between $p$ and $p_\ell$.)
			Let $w'$ be the resultant restricted weight function. \\

			\BlankLine

			With ties broken arbitrarily, select any $k+1$ minimum weighted points in $A$, with respect to weights associated via $w'$; let $A'$ be this set of points. \\

			\BlankLine

			For every $p \in A \cup B$ and $q \in A'$, add edge $(p,q)$ to $G$.
		}
	}
\end{algorithm}
        
Essentially, edges added to $G$ in Algorithm~\ref{algo:vftswpsimppoly} help in achieving the vertex fault-tolerance, i.e., maintaining $\sqrt{10}+\epsilon$ stretch even after removing any $k$ points in $S$.
We restate the following lemma from \cite{conf/compgeom/AbamAHA15}, which is useful in the analysis of Algorithm~\ref{algo:vftswpsimppoly}.

\begin{lemma}
\label{lem4}
(from \cite{conf/compgeom/AbamAHA15})
Suppose $ABC$ is a right triangle with $\angle CAB = \frac{\pi}{2}$.
For some point $D$ on line segment $AC$, let $H$ be a $y$-monotone path between $B$ and $D$ such that the region bounded by $AB, AD$, and $H$ is convex.
Then, $3d(H) + d(D,C) \leq \sqrt{10}d(B,C)$, where $d(., .)$ denotes the Euclidean length.
\end{lemma}

We note that a $y$-monotone path is a path whose intersection with any line perpendicular to $y$-axis is connected. 
In the following lemma, we prove the graph $G$ constructed in Algorithm~\ref{algo:vftswpsimppoly} is indeed a $(k, \sqrt{10}+\epsilon, w)$-VFTSWP for the set $S$ of points located in $\cal P$.

\begin{lemma}
\label{lem5}
The spanner $G$ computed by Algorithm~\ref{algo:vftswpsimppoly} is a geodesic $(k, \sqrt{10}+\epsilon, w)$-VFTSWP for the set $S$ of points located in $\cal P$.
\end{lemma}
\begin{proof}
Consider any set $S' \subset S$ such that $|S'| \leq k$. 
Let $p, q$ be any points in $S \setminus S'$. 
First, we note that there exists a splitting line segment $\ell$ at some iteration of the algorithm such that $p$ and $q$ lie on different sides of $\ell$.
Let $\pi(p, q)$ be a shortest path between $p$ and $q$.
Also, let $r$ be a point at which $\pi(p,q)$ intersects $\ell$.
Consider a pair $(A,B)$ in $\frac{4}{\epsilon}$-SSPD $\cal S$ such that $p_\ell \in A$ and $q_\ell \in B$ or, $q_\ell \in A$ and $p_\ell \in B$. 
We note that since $\cal S$ is a pair decomposition of points in $S_l$, such a $(A, B)$ pair always exists in $\cal S$.
Without loss of generality, assume the former holds.
When $|A| < k+1$, there exists an edge between $p$ and $q$ in $G$. 
Hence, $d_{G \setminus S'}(p,q) = d_w(p,q)$.
Consider the other case in which $|A| \geq k+1$.
Since $|S'| \le k$, there exists a $c_j \in A$ such that $c_j \notin S'$.
Therefore, $d_{G \setminus S'}(p,q)$
{\setlength{\abovedisplayskip}{0pt}
\begin{flalign}
\hspace{0.2mm}
    &=d_w(p, c_j) + d_w(c_j,q)&&\nonumber\\
    &= w(p) + d_{\pi}(p,c_j) + w(c_j) + w(c_j) + d_{\pi}(c_j,q) + w(q)&&\nonumber\\
    &\hspace{5mm}\text{[by the definition of $A'$]}&&\nonumber\\
    &\leq w(p) + d_{\pi}(p,p_\ell) + |p_\ell c_{j_\ell}| + d_\pi(c_{j_\ell},c_j) + w(c_j) + w(c_j) + &&\nonumber\\
    &\hspace{5mm} d_{\pi}(c_{j_\ell},c_j) + |c_{j_\ell} q_\ell| + d_{\pi}(q_\ell,q) + w(q)&&\nonumber\\
    &\hspace{5mm}\text{[since geodesic shortest paths follow triangle inequality]}&&\nonumber\\
	&\leq w(p) + d_{\pi}(p,p_\ell) + |p_\ell c_{j_\ell}| + w(p) + w(p) + d_{\pi}(p,p_\ell) + d_{\pi}(p,p_\ell) &&\nonumber\\
    &\hspace{5mm}+ |c_{j_\ell} q_\ell| + d_{\pi}(q_\ell,q) + w(q)&&\nonumber\\
    &\hspace{5mm}\text{[by the definition of $A'$]}&&\nonumber\\
    &\leq 3[w(p) + w(q)] + 3d_{\pi}(p,p_\ell) + |p_\ell c_{j_\ell}| + |c_{j_\ell} q_\ell| + d_{\pi}(q_\ell,q)&&\nonumber\\
    &\leq 3[w(p) + w(q)] + 3d_{\pi}(p,p_\ell) + |p_\ell c_{j_\ell}| + |p_\ell c_{j_\ell}| + |p_\ell r| + |r q_\ell| + d_{\pi}(q_\ell,q)&&\nonumber\\
    &\hspace{5mm}\text{[since Euclidean distances follow triangle inequality]}&&\nonumber\\
    &= 3[w(p) + w(q)] + 3d_{\pi}(p,p_\ell) + 2|p_\ell c_{j_\ell}| + |p_\ell r| + |r q_\ell| + d_{\pi}(q_\ell,q)&&\nonumber\\
    &\leq 3[w(p) + w(q)] + 3d_{\pi}(p,p') + 3d_{\pi}(p',p_\ell) + 2|p_\ell c_{j_\ell}| + |p_\ell r| + &\nonumber\\
    &\hspace{5mm}|r q_\ell| + d_{\pi}(q_\ell,q') + d_{\pi}(q',q)&&\nonumber\\
    &\hspace{5mm}\text{[by the triangle inequality of Euclidean distances; here, $p'$ (resp. $q'$) is the} &\nonumber\\
    &\hspace{5mm}\text{the first point at which $\pi(p,q)$ and $\pi(p,p_\ell)$ (resp. $\pi(q,q_\ell)$) part ways]}&&\nonumber\\
    &\leq 3[w(p) + w(q)] + 3d_{\pi}(p,p') + \sqrt{10}d_{\pi}(p',r) + 2|p_\ell c_{j_\ell}| + \sqrt{10}d_{\pi}(q',r) +&&\nonumber\\
    &\hspace{5mm}d_{\pi}(q',q)&&\nonumber\\
    &\hspace{5mm}\text{[applying Lemma~\ref{lem4} to triangles $q' h_{q'} r$ and $p' h_{p'} r$, where $h_{p'}$ (resp. $h_{q'}$)}&&\nonumber\\
    &\hspace{5mm}\text{is the projection on to line defined by $p_\ell$ ($q_\ell$)and $r$]}&\nonumber\\
    &\leq 3[w(p) + w(q)] + 3d_{\pi}(p,p') + \sqrt{10}d_w(p',r) + 2|p_\ell c_{j_\ell}| + \sqrt{10}d_w(q',r) +&&\nonumber\\
    &\hspace{5mm}d_{\pi}(q',q)&&\nonumber\\
    &\hspace{5mm}\text{[since the weight associated with any point is non-negative]}&&\nonumber\\
    &\leq 3[w(p) + w(q)] + 3d_{\pi}(p,p') + \sqrt{10}d_{w}(p',q') + 2|p_\ell c_{j_\ell}| + d_{\pi}(q',q)\label{eq14}&& \\
	&\hspace{5mm}\text{[$r$ is the point where $\pi(p,q)$ intersects $\ell$; optimal substructure}&&\nonumber
\end{flalign}}

{\setlength{\abovedisplayskip}{0pt}
\begin{flalign}
\hspace{0.2mm}
	&\hspace{5mm}\text{property of shortest paths says $d_w(p,q) = d_w(p,r) + d_w(r,q)$]}&&\nonumber\\
	&\leq 3[w(p) + w(q)] + 3d_{\pi}(p,p') + \sqrt{10}d_w(p',q') + \epsilon d_w(p,q) + d_{\pi}(q',q). \label{eqbr}&&
\end{flalign}}

\noindent
Since $\cal S$ is a $\frac{4}{\epsilon}$-SSPD for the set $S_\ell$ of points, for any pair $(X, Y)$ of $\cal S$, the distance between any two points in $X$ is at most $\frac{\epsilon}{2}$ times of the distance between $X$ and $Y$. 
Hence,
{\setlength{\abovedisplayskip}{0pt}
\begin{flalign}
\hspace{6mm}&|p_\ell c_{j_\ell}| \leq \frac{\epsilon}{2}|p_\ell q_\ell|.\label{eq15}&
\end{flalign}}

\noindent
Therefore, $|p_\ell q_\ell|$
{\setlength{\abovedisplayskip}{0pt}
\begin{flalign}
    &\leq |p_\ell r| + |r q_\ell|&&\nonumber\\
    &\hspace{5mm}\text{[by the triangle inequality]}&&\nonumber\\
        &\leq |p_\ell p| + |p r| + |r q| + |q q_\ell|&&\nonumber\\
    &\hspace{5mm}\text{[by the triangle inequality]}&&\nonumber\\
        &\leq d_{\pi}(p_\ell,p) + d_{\pi}(p,r) + d_{\pi}(r,q) + d_{\pi}(q,q_\ell)&&\nonumber\\
        &\leq d_{\pi}(p,r) + d_{\pi}(p,r) + d_{\pi}(r,q) + d_{\pi}(r,q)&&\nonumber\\
    &\hspace{5mm}\text{[by the definition of projection of a point on} \ l \ \text{]}&&\nonumber\\
    &\leq d_w(p,r) + d_w(p,r) + d_w(r,q) + d_w(r,q)&&\nonumber\\
    &\hspace{5mm}\text{[since the weight associated with each point is non-negative]}&&\nonumber\\
        &= 2d_w(p,q)\label{eq16}&&\\
    &\hspace{5mm}\text{[since} \ r \ \text{is the point where} \ \pi(p,q) \ \text{intersects} \ l \ \text{].}&\nonumber
\end{flalign}}

\noindent
From (\ref{eq15}) and (\ref{eq16}),
{\setlength{\abovedisplayskip}{0pt}
\begin{flalign}
	&\hspace{5mm}|p_\ell c_{j_\ell}| \leq \epsilon d_w(p,q).\label{eqplell1} &&
\end{flalign}}

\noindent
Then, $d_{G \setminus S'}(p,q)$
{\setlength{\abovedisplayskip}{0pt}
\begin{flalign}
	&\leq 3[w(p) + w(q)] + 3d_{\pi}(p,p') + \sqrt{10}d_w(p',q') + \epsilon d_w(p,q) + d_{\pi}(q',q). &&\nonumber\\
	&\hspace{5mm}\text{[from (\ref{eqbr}) and (\ref{eqplell1})]} && \nonumber \\
        &\leq 3d_w(p,p') + \sqrt{10}d_w(p',q') + \epsilon d_w(p,q) + d_w(q',q)&&\nonumber\\
	&\hspace{5mm}\text{[since the weight associated with each point is non-negative]}&&\nonumber\\
        &\leq \sqrt{10}d_w(p,q) + \epsilon d_w(p,q)&&\nonumber\\
	&\hspace{5mm}\text{[since $d_\pi(p,q) = d_\pi(p,p') + d_\pi(p',q') + d_\pi(q',q)$]}.&&\nonumber
\end{flalign}}
Hence, $d_{G \setminus S'}(p,q) \le (\sqrt{10} + \epsilon)d_w(p,q)$.
\end{proof}

\begin{theorem}
\label{thm:simppoly}
Given a simple polygon $\cal P$, a set $S$ of $n$ points located in $\cal P$, a weight function $w$ to associate a non-negative weight to each point in $S$, a positive integer $k$, and a real number $0 < \epsilon \le 1$, Algorithm~\ref{algo:vftswpsimppoly} computes a $(k, \sqrt{10}+\epsilon, w)$-vertex fault-tolerant geodesic spanner network $G$ with $O(k n (\lg n)^{2})$ edges, for the set $S$ of weighted points.
\end{theorem}
\begin{proof}
From Lemma~\ref{lem5}, the spanner constructed is $(k, \sqrt{10}+\epsilon, w)$-VFTSWP. 
Let $S(n)$ be the number of edges in the spanner.
Also, let $n_1, n_2$ be the sizes of sets obtained by partitioning the initial $n$ points at the root node of the divide-and-conquer recursion tree. 
The recurrence is $S(n) = S(n_1) + S(n_2) + \sum_{(A, B) \in \frac{4}{\epsilon}\text{-SSPD}} (k(|A| + |B|))$.
Since $(A, B)$ is a pair in $\frac{4}{\epsilon}$-SSPD, $\sum (|A| + |B|) = O(n \lg n) $.
Noting that $n_1, n_2 \ge n/3$, the size of the spanner is $O(k n (\lg n)^{2})$.
\end{proof}

\section{Vertex fault-tolerant spanner for weighted points in a polygonal domain}
\label{sect:polydom}

Given a polygonal domain $\cal P$, a set $S$ of $n$ points located in the free space $\cal D$ of $\cal P$, a weight function $w$ to associate a non-negative weight to each point in $S$, a positive integer $k$, and a real number $0 < \epsilon \le 1$, we compute a geodesic $(k, 6+\epsilon, w)$-VFTSWP for the set $S$ of weighted points. 
That is, for any set $S' \subset S$ with cardinality at most $k$, the graph $G \backslash S'$ is a $(6+\epsilon)$-spanner for the set $S'$ of weighted points. 

Algorithm~\ref{algo:vftswpolydom} mentioned below computes a geodesic $(6+\epsilon)$-VFTSWP $G$ for a set $S$ of $n$ points lying in the free space $\cal D$ of the polygonal domain $\cal P$, while each point in $S$ is associated with a non-negative weight.
The polygonal domain $\cal P$ consists of a simple polygon and $h$ simple polygonal holes located interior to that simple polygon.
A {\it splitting line segment} is a line segment in the free space with both of its endpoints on the boundary of the polygonal domain.
(Note that the boundary of a polygonal domain is the union of boundaries of holes and the boundary of the outer polygon.)
We combine and extend the algorithms given in Section~\ref{sect:simppoly} and the algorithms in \cite{conf/compgeom/AbamAHA15} to the case of input points associated with non-negative weights.
In addition, the spanner constructed by this algorithm is vertex fault-tolerant and achieves a $6+\epsilon$ multiplicative stretch.

The following theorem from \cite{journals/siamdm/AlonST94} on computing a planar separator helps in devising a divide-and-conquer based algorithm (refer to Algorithm~\ref{algo:vftswpolydom}).

\begin{theorem}
\label{thm4}
(\cite{journals/siamdm/AlonST94})
Suppose $G=(V,E)$ is a planar vertex-weighted graph with $|V| = m$. 
Then, an $O(\sqrt{m})$-separator for $G$ can be computed in $O(m)$ time. That is, $V$ can be partitioned into sets $P, Q,$ and $R$ such that $|R| = O(\sqrt{m})$, there is no edge between $P$ and $Q$, and $w(P),w(Q) \leq \frac{2}{3}w(V)$.
Here, $w(X)$ is the sum of weights of all vertices in $X$.
\end{theorem}

\begin{algorithm}
	\caption{VFTSWPPolygonalDomain(${\cal D}, S$)}
	\label{algo:vftswpolydom}
    
	\SetAlgoLined
	\SetKwInOut{Input}{Input}
	\SetKwInOut{Output}{Output}
	\Input{The free space $\cal D$ of a polygonal domain $\cal P$, a set $S$ on $n$ points located in $\cal D$, a weight function $w$ that associates a non-negative weight to each point in $S$, an integer $k \geq 1$, and a real number $0 < \epsilon \le 1$.}
	\Output{A $(k, 6+\epsilon, w)$-VFTSWP $G$.}

	\BlankLine
    
	If $|S| \le 1$ then return. \\

	\BlankLine

	Partition the free space $\cal D$ into $O(h)$ simple polygons using $O(h)$ splitting line segments, such that no splitting line segment crosses any of the holes of $\cal P$, and each of the resultant simple polygons has at most three splitting line segments bounding it.
	This is done by choosing a leftmost vertex $l_O$ (resp. a rightmost vertex $r_O$) along the $x$-axis of each obstacle $O \in \cal P$, and projecting $l_O$ (resp. $r_O$) both upwards and downwards, parallel to $y$-axis.
	(If any of the resultant simple polygons have more than three splitting line segments on its boundary, then that polygon is further decomposed arbitrarily so that the resulting polygon has at most three splitting line segments on its boundary.) \\

	\BlankLine

	A planar graph $G_d$ is constructed: each vertex of $G_d$ corresponds to a simple polygon in the above decomposition; for any two vertices $v', v''$ of $G_d$, an edge $(v', v'')$ is introduced into $G_d$ whenever the corresponding simple polygons of $v'$ and $v''$ are adjacent to each other in the decomposition.
	Further, each vertex $v$ of $G_d$ is associated with a weight equal to the number of points that lie inside the simple polygon corresponding to $v$. \\

	\BlankLine

	\If{the number of vertices in $G_d$ is $1$}{
		Let $\cal D'$ be the simple polygon that corresponds to the vertex of $G_d$, and let $S'$ be the points in $S$ that belong to $\cal D'$. \\

		\BlankLine

		VFTSWPSimplePolygon($\cal D', S'$). \scriptsize{} (Refer to Algorithm~\ref{algo:vftswpsimppoly}.) \normalsize{} \\

		\BlankLine

		\Return
	}

	\BlankLine

	Compute an $O(\sqrt{h})$-separator $R$ for the planar graph $G_d$ using Theorem~\ref{thm4}, and let $P, Q,$ and $R$ be the sets into which the vertices of $G_d$ is partitioned. \\

	\BlankLine

	For each vertex $r \in R$, we collect the bounding splitting line segments of the simple polygon corresponding to $r$ into $H$, i.e., $O(\sqrt{h})$ splitting line segments are collected into a set $H$. \\
            
	\BlankLine

	\ForEach{$l \in H$}{
		For each point $p$ that lies in $\cal D$, compute a geodesic projection $p_\ell$ of $p$ on $\ell$. 
		Let $S_\ell$ be the set of points resultant from these projections. \\

		\BlankLine

		Using the algorithm from \cite{journals/dcg/ABFG09}, compute a $\frac{8}{\epsilon}$-SSPD $\cal S$ for the points in $S_\ell$. \\

		\BlankLine

		IncludeEdgesUsingSSPD(${\cal S}, G$).  \scriptsize{} (Refer to Algorithm~\ref{algo:addedgessspd}.) \normalsize{}
	}
    
	Let $\cal D'$ (resp. $\cal D''$) be the union of simple polygons that correspond to vertices of $G_d$ in $P$ (resp. $Q$).
	(Note that $\cal D'$ and $\cal D''$ are polygonal domains.)
	Also, let $S'$ (resp. $S''$) be the points in $S$ that belong to $\cal D'$ (resp. $\cal D''$). \\

	\BlankLine

	VFTSWPPolygonalDomain(${\cal D'}, S'$). \\

	\BlankLine

	VFTSWPPolygonalDomain(${\cal D''}, S''$).

\end{algorithm}

The following lemma proves that $G$ computed by Algorithm~\ref{algo:vftswpolydom} is a geodesic $(k, 6+\epsilon, w)$-VFTSWP for the set $S$ of points.

\begin{lemma}
\label{lem17}
The spanner $G$ is a geodesic $(k, 6+\epsilon, w)$-VFTSWP for the set $S$ of points located in $\cal D$.
\end{lemma}
\begin{proof}
Consider any set $S' \subset S$ such that $|S'| \leq k$.
Let $p, q$ be any two points in $S \setminus S'$. 
Based on the locations of $p$ and $q$, we consider the following cases:
(1) The points $p$ and $q$ lie inside the same simple polygon and no shortest path between $p$ and $q$ intersects any splitting line segment from the set $H$.
(2) The points $p$ and $q$ belong to two distinct simple polygons in the simple polygonal subdivision of $\cal D$, and both of these simple polygons correspond to vertices of one of $P, Q,$ and $R$.
(3) The points $p$ and $q$ belong to two distinct simple polygons in the simple polygonal subdivision of $\cal D$, but if one of these simple polygons correspond to a vertex of $P' \in \{P, Q, R\}$, then the other simple polygon correspond to a vertex of $P'' \in \{P, Q, R\}$ and $P'' \ne P'$.
In Case (1), we run Algorithm~\ref{algo:vftswpsimppoly}, which implies there exists a path with $\sqrt{10} + \epsilon$ multiplicative stretch between $p$ and $q$.
In Cases (2) and (3), a shortest path between $p$ and $q$ intersects at least one of the $O(\sqrt{h})$ splitting line segments collected in the set $H$.
Let $\ell$ be the splitting line segment that intersects a shortest path, say $\pi(p, q)$, between $p$ and $q$.
Also, let $r$ be the point of intersection of $\pi(p, q)$ with $\ell$.
At this recursive step, consider a pair $(A,B)$ in $\frac{8}{\epsilon}$-SSPD $\cal S$ such that $p_\ell \in A$ and $q_\ell \in B$ or, $q_\ell \in A$ and $p_\ell \in B$, where $p_\ell$ (resp. $q_\ell$) is the projection of $p$ (resp. $q$) on $\ell$. 
Without loss of generality, assume the former holds.
Suppose $|A| < k+1$.
Then, there exists an edge between $p$ and $q$.
Hence, $d_{G \setminus S'}(p,q) = d_w(p,q)$.
Consider the other case in which $|A| \geq k+1$.
Since $|S'| \le k$, there exists a $c_j \in A$ such that $c_j \notin S'$.
Therefore, 
{\setlength{\abovedisplayskip}{0pt}
\begin{flalign}
	&d_{G \setminus S'}(p,q)&&\nonumber\\
    &=d_w(p, c_j) + d_w(c_j,q)&&\nonumber\\
    &= w(p) + d_{\pi}(p,c_j) + w(c_j) + w(c_j) + d_{\pi}(c_j,q) + w(q).\label{eq27}& \\
    &\leq w(p) + d_{\pi}(p,p_\ell) + |p_\ell c_{j_\ell}| + d_{\pi}(c_{j_\ell},c_j) + w(c_j) + w(c_j) + d_{\pi}(c_{j_\ell},c_j) + &&\nonumber\\
    &\hspace{5mm} |c_{j_\ell} q_\ell| + d_{\pi}(q_\ell,q) + w(q) &&\nonumber\\
    &\hspace{5mm}\text{[since geodesic shortest paths follow triangle inequality]}&&\nonumber\\
    &\leq w(p) + d_w(p,q) + |p_\ell c_{j_\ell}| + d_{\pi}(c_{j_\ell},c_j) + w(c_j) + w(c_j) + d_{\pi}(c_{j_\ell},c_j) +&&\nonumber\\
	&\hspace{5mm}|c_{j_\ell} q_\ell| + w(q) \hspace{5mm}
	\text{[as $d_{\pi}(p,p_\ell) + d_{\pi}(q,q_\ell) \le d_w(p,q)$]}&&\nonumber\\
	&\leq w(p) + d_w(p,q) + \frac{\epsilon}{4}|p_\ell q_\ell| + 2 d_{\pi}(c_{j_\ell},c_j) + 2w(c_j) + &&\nonumber\\
	&\hspace{5mm}|c_{j_\ell} q_\ell| + w(q) \hspace{5mm}
	\text{[as $\cal S$ is a $\frac{8}{\epsilon}$-SSPD, $|p_\ell c_{j_\ell}| \leq \frac{\epsilon}{4}|p_\ell q_\ell|$]}&&\nonumber \\
	&\leq w(p) + d_w(p,q) + \frac{\epsilon}{4}[|p_\ell p| + |pq| + |q q_\ell|] + 2 d_{\pi}(c_{j_\ell},c_j) + 2w(c_j) + &&\nonumber \\
	&\hspace{5mm} |c_{j_\ell} q_\ell| + w(q) \hspace{5mm}\text{[by the triangle inequality]}&&\nonumber\\
    &\leq w(p) + d_w(p,q) + \frac{\epsilon}{4}[d_{\pi}(p_\ell,p) + d_{\pi}(p,q) + d_{\pi}(q,q_\ell)] + 2 d_{\pi}(c_{j_\ell},c_j) + 2w(c_j)&&\nonumber\\
	&\hspace{5mm}+ |c_{j_\ell} q_\ell| + w(q) \hspace{5mm}\text{[as the Euclidean distance between any two points}&&\nonumber\\
	&\hspace{5mm}\text{cannot be greater to the geodesic distance between them]}&&\nonumber\\
	&\leq w(p) + d_w(p,q) + \frac{\epsilon}{4}[d_w(p,q) + d_{\pi}(p,q)] + 2 d_{\pi}(c_{j_\ell},c_j) + 2w(c_j) +&&\nonumber\\
	&\hspace{5mm}|c_{j_\ell} q_\ell| + w(q) \hspace{5mm}\text{[as $\cal S$ is a $\frac{8}{\epsilon}$-SSPD].}\label{eqbr2}&&
\end{flalign}}

\noindent
Since $\cal S$ is a $\frac{8}{\epsilon}$-SSPD for the set $S_\ell$ of points, for any pair $(X, Y)$ of $\cal S$, the distance between any two points in $X$ is at most $\frac{\epsilon}{4}$ times of the distance between $X$ and $Y$. 
Therefore,
{\setlength{\abovedisplayskip}{0pt}
\begin{flalign}
\hspace{6mm}|p_\ell c_{j_\ell}| \leq \frac{\epsilon}{4}|p_\ell q_\ell|.\label{eq35}&&
\end{flalign}

\noindent
Since $r$ is the point where $\pi(p,q)$ intersects $\ell$, by the optimal sub-structure property of shortest paths,
\begin{flalign}
\hspace{6mm}&\pi(p,q) = \pi(p,r) + \pi(r,q).\label{eq36}&
\end{flalign}}

\noindent
Then, $d_{\pi}(p,p_\ell) + d_{\pi}(q,q_\ell)$ 
{\setlength{\abovedisplayskip}{0pt}
\begin{flalign}
    &\leq d_{\pi}(p,r) + d_{\pi}(q,r)&&\nonumber\\
    &\hspace{5mm}\text{[as $d_{\pi}(p, p_\ell) \leq d_{\pi}(p, r)$ and $d_{\pi}(q, q_\ell) \leq d_{\pi}(q, r)$]}&&\nonumber\\
        &= d_{\pi}(p,q)&&\nonumber\\
	&\hspace{5mm}\text{[since $\pi(p, q)$ intersects $\ell$ at $r$]}&&\nonumber\\
        &\leq d_w(p,q)\label{eq37}&&\\
    &\hspace{5mm}\text{[since the weight associated with each point is non-negative].}&\nonumber
\end{flalign}}

\noindent
Moreover, $|c_{j_\ell} q_\ell|$
{\setlength{\abovedisplayskip}{0pt}
\begin{flalign}
    &\leq |c_{j_\ell} p_\ell| + |p_\ell q_\ell|&&\nonumber\\
    &\hspace{5mm}\text{[by the triangle inequality]}&&\nonumber\\
        &\leq \frac{\epsilon}{4}|p_\ell q_\ell| + |p_\ell q_\ell|&&\nonumber\\
	&\hspace{5mm}\text{[from} \ (\ref{eq35})\text{]}&&\nonumber\\
        &\leq (\frac{\epsilon}{4} + 1)|p_\ell q_\ell|&&\nonumber\\
        &\leq (\frac{\epsilon}{4} + 1)(|p_\ell r| + |r q_\ell|)&&\nonumber \\
    &\hspace{5mm}\text{[by the triangle inequality]}&&\nonumber\\
        &\leq (\frac{\epsilon}{4} + 1)(|p_\ell p| + |pq| + |q q_\ell|)&&\nonumber\\
    &\hspace{5mm}\text{[by the triangle inequality]}&&\nonumber\\
        &\leq (\frac{\epsilon}{4} + 1)(d_{\pi}(p_\ell,p) + d_{\pi}(p,q) + d_{\pi}(q,q_\ell))&&\nonumber\\
    &\hspace{5mm}\text{[from the definition of geodesic distance]}&&\nonumber\\
        &\leq (\frac{\epsilon}{4} + 1)(d_{\pi}(r,p) + d_{\pi}(p,q) + d_{\pi}(q,r))&&\nonumber\\
    &\hspace{5mm}\text{[as $d_{\pi}(p, p_\ell) \leq d_{\pi}(p, r)$ and $d_{\pi}(q, q_\ell) \leq d_{\pi}(q, r)$]}&&\nonumber\\
        &= (\frac{\epsilon}{4} + 1)(d_w(p,q) + d_w(p,q))&&\nonumber \\
    &\hspace{5mm}\text{[from} \ (\ref{eq37})\text{]}&&\nonumber\\
	&= 2(\frac{\epsilon}{4} + 1)d_w(p,q).&\nonumber
\end{flalign}}

\noindent
Hence, 
{\setlength{\abovedisplayskip}{0pt}
\begin{flalign}
	&|p_\ell c_{j_\ell}| \le 2(\frac{\epsilon}{4} + 1)d_w(p,q). \label{eqbr3}&&
\end{flalign}}

{\setlength{\abovedisplayskip}{0pt}
\begin{flalign}
	&d_{G \setminus S'}(p,q) &&\nonumber\\
	&\leq w(p) + d_w(p,q) + \frac{2 \epsilon}{4}d_w(p,q) + 2 d_{\pi}(c_{j_\ell},c_j) + 2w(c_j) + (\frac{\epsilon}{2} + 2)d_w(p,q) + &&\nonumber\\
	&\hspace{5mm}w(q) \hspace{5mm}\text{[from (\ref{eqbr2}) and (\ref{eqbr3})]}&&\nonumber\\
	&\leq w(p) + d_w(p,q) + \frac{2 \epsilon}{4}d_w(p,q) + 2[w(p) + d_{\pi}(p,q)] + (\frac{\epsilon}{2} + 2)d_w(p,q) + &&\nonumber\\
	&\hspace{5mm}w(q) \hspace{5mm}\text{[by the definition of set $A'$]}&&\nonumber\\
	&\leq 3[w(p) + d_{\pi}(p,q) + w(q)] + d_w(p,q) + \frac{2 \epsilon}{4}d_w(p,q) + (\frac{\epsilon}{2} + 2)d_w(p,q)&&\nonumber\\
	&= 3d_w(p,q) + d_w(p,q) + \frac{2 \epsilon}{4}d_w(p,q) + (\frac{\epsilon}{2} + 2)d_w(p,q).&&\nonumber
\end{flalign}}
\vspace{-0.1in}
Hence, $d_{G \setminus S'}(p,q) \leq (6 + \epsilon)d_w(p,q)$.
\end{proof}

\begin{theorem}
\label{thm:polydom}
Given a polygonal domain $\cal P$, a set $S$ of $n$ points located in the free space $\cal D$ of $\cal P$, a weight function $w$ to associate a non-negative weight to each point in $S$, a positive integer $k$, and a real number $0 < \epsilon \le 1$, Algorithm~\ref{algo:vftswpolydom} computes a $(k, 6+\epsilon, w)$-vertex fault-tolerant spanner network $G$ with $O(k n \sqrt{h+1} (\lg n)^2)$ edges, for the set $S$ of weighted points. 
\end{theorem}
\begin{proof}
From Lemma~\ref{lem17}, the spanner constructed is $(k, 6+\epsilon, w)$-VFTSWP.
Let $S(n)$ be the number of edges in $G$.
Since there are $O(\sqrt{h})$ splitting line segments, the number of edges included into $G$ is $O(k n \sqrt{h} \lg{n})$.
At each internal node of the recursion tree of the algorithm, $S(n)$ satisfies the recurrence $S(n) = O(k n \sqrt{h} \lg{n}) + S(n_1) + S(n_2)$, where $n_1 + n_2 = n$ and $n_1, n_2 \geq n/3$.
The number of edges included at all the leaves of the recursion tree together is $O(k n \sqrt{h} (\lg n)^{2})$.
Hence, the total number of edges is $O(k n \sqrt{h} (\lg n)^{2})$.
(The plus $1$ under the square root is for considering the case in which the input polygonal domain has no holes, i.e., when the input polygonal domain is a simple polygon.)
\end{proof}

\section{Vertex fault-tolerant spanner for weighted points located on a polyhedral terrain}
\label{sect:terr}

Given a polyhedral terrain $\cal T$, a set $S$ of $n$ points located on $\cal T$, a weight function $w$ to associate a non-negative weight to each point in $S$, a positive integer $k$, and a real number $0 < \epsilon \le 1$, in this section we describe an algorithm to compute a geodesic $(k, 6+\epsilon)$-vertex fault-tolerant spanner with $O(kn\sqrt{n}(\lg{n})^2)$ edges for the set $S$ of weighted points.

We denote the boundary of $\mathcal{T}$ with $\partial \mathcal{T}$.
Also, we denote a geodesic Euclidean shortest path (a path lying on $\cal T$) between any two points $a, b \in \mathcal{T}$ with $\pi(a, b)$.
The distance along $\pi(a, b)$ is denoted by $d_\pi(a, b)$.
For any two points $x, y \in \partial \mathcal{T}$, we denote the closed region lying to the right (resp. left) of $\pi(x,y)$ when going from $x$ to $y$, including (resp. excluding) the points lying on the shortest path $\pi(x,y)$ with $\pi^{+}(x,y)$ (resp. $\pi^{-}(x,y)$).
The {\it geodesic projection} $p_\pi$ of a point $p \in {\cal T}$ on a shortest path $\pi$ between two points lying on $\mathcal{T}$ is defined as a point on $\pi$ that is at the minimum geodesic distance from $p$ among all the points of $\pi$.

For three points $u,v,w \in \mathcal{T}$, the closed region bounded by shortest paths $\pi(u,v)$, $\pi(v,w)$, and $\pi(w,u)$ is called an {\it sp-triangle}, denoted by $\Delta(u, v, w)$.
If points $u, v, w \in \mathcal{T}$ are clear from the context, we denote the sp-triangle with $\Delta$.
In the following, we restate a theorem from \cite{conf/soda/AbamBS17}, which is useful for our analysis.

\begin{theorem}
\label{thm7}
For any set $P$ of $n$ points on a polyhedral terrain $\mathcal{T}$, there exists a balanced sp-separator: a shortest path $\pi(u,v)$ connecting two points $u,v \in \partial \mathcal{T}$ such that $\frac{2n}{9} \leq |\pi^{+}(u,v) \cap P| \leq \frac{2n}{3}$, or a sp-triangle $\Delta$ such that $\frac{2n}{9} \leq |\Delta \cap P| \leq \frac{2n}{3}$.
\end{theorem}

Thus, an sp-separator is either bounded by a shortest path (in the former case), or by three shortest paths (in the latter case).
We present two divide-and-conquer based algorithms to compute VFTSWP for weighted points located on terrain.

As part of our first algorithm, a balanced sp-separator as given in Theorem~\ref{thm7} is computed. 
Let $\gamma$ be a shortest path that belongs to an sp-separator.
The sets $S_{in}$ and $S_{out}$ comprising points are defined as follows: if the sp-separator is a shortest path, then define $S_{in}$ to be $\gamma^{+}(u,v) \cap S$; otherwise, $S_{in}$ is $\Delta \cap S$; and, the set $S_{out}$ is defined as $S - S_{in}$.
For each point $p \in S$, we compute a projection $p_{\gamma}$ of $p$ on every shortest path $\gamma$ of sp-separator, and associate a weight $w(p) + d_{\mathcal{T}}(p, p_{\gamma})$ with $p_{\gamma}$.
Let $S_{\gamma}$ be the set $\cup_{p \in S}\hspace{0.02in} p_{\gamma}$.
For every point $p_\gamma \in S_\gamma$, we associate weight $w(p) + d_\pi(p, p_\gamma)$ to $p_\gamma$, where $p_\gamma$ is a geodesic projection of $p$ on to $\gamma$.
Suppose $G_\gamma$ is a $(k, t)$-VFTS for the weighted points located on $\gamma$.
For every edge $(p_{\gamma}, q_{\gamma}) \in G_{\gamma}$, we introduce an edge $(p,q)$ into $G$, where $p_{\gamma}$ (resp. $q_{\gamma}$) is the projection of $p$ (resp. $q$) on $\gamma$. 
More edges are added to $G$ while recursively processing points in sets $S_{in}$ and $S_{out}$.
Refer to Algorithm~\ref{algo:terr1}.

\begin{algorithm}[ht]
    \caption{VFTSWPTerrain1(${\cal T}, S, k, \epsilon$)}
    \label{algo:terr1}

    \SetKwInOut{Input}{Input}
    \SetKwInOut{Output}{Output}
    
    \Input{A triangulated polyhedral terrain $\cal T$, a set $S$ on $n$ points located on $\cal T$, a weight function $w$ that associates a non-negative weight to each point in $S$, an integer $k \geq 1$, and a real number $0 < \epsilon \le 1$.}
    \Output{A $(k, 3t+\epsilon, w)$-VFTSWP $G$, where $t$ is the spanning ratio of $G_\gamma$.}
    
	\BlankLine

    \While{$|\mathcal{T} \cap S| \geq 1$}{
        
	\BlankLine

        Compute a balanced sp-separator $\Gamma$ for $\mathcal{T}$ using the algorithm given in \cite{conf/soda/AbamBS17}. \\
        
	\BlankLine

        \ForEach{shortest path $\gamma$ that is bounding $\Gamma$}{
			
		\BlankLine

		For each point $p$ that lies on $\cal T$, compute a geodesic projection $p_\gamma$ of $p$ on $\gamma$ and associate weight $w(p)+d_\pi(p, p_\gamma)$ to $p_\gamma$. 
		Let $S_\gamma$ be the set of points resultant from these projections. \\

		\BlankLine

		Compute a $(k, t)$-VFTS $G_\gamma$ for the set $S_\gamma$ of weighted points in $S$.

		\BlankLine

		For every edge $(p_\gamma, q_\gamma) \in G_\gamma$, add an edge $(p, q)$ to $G$, where $p_\gamma$ (resp. $q_\gamma$) is a geodesic projection of $p$ (resp. $q$) on $l$.
on l.
		\BlankLine

        }

	\BlankLine

        \scriptsize{}
        Let ${\cal T}'$ be $\pi^{+}(u,v)$ if the balanced sp-separator is a shortest path $\pi(u,v)$; otherwise, let ${\cal T}'$ be $\Delta$.
        Also, let $S_{in}$ be the set of points in $S$ located on the polyhedral terrain ${\cal T}'$. \\
        \normalsize{}

	\BlankLine

        VFTSWPTerrain(${\cal T}', S_{in}, k, \epsilon$). \\
    
	\BlankLine

        \scriptsize{}
        Let ${\cal T}''$ be $\pi^{-}(u,v)$ if the balanced sp-separator is a shortest path $\pi(u,v)$; otherwise, let ${\cal T}''$ be ${\cal T} \setminus \Delta$.
        Also, let $S_{out}$ be the set of points in $S$ located on the polyhedral terrain ${\cal T}''$. \\
        \normalsize{}

	\BlankLine

        VFTSWPTerrain(${\cal T}'', S_{out}, k, \epsilon$). 

	\BlankLine

    }
        
\end{algorithm}

\begin{lemma}
\label{lem20}
Given a set $S$ of $n$ points located on a polyhedral terrain $\cal T$, a weight function $w$ to associate a non-negative weight to each point in $S$, a positive integer $k$, and a real number $\epsilon > 0$, Algorithm~\ref{algo:terr1} computes a $(k, 3t)$-vertex fault tolerant geodesic spanner for the weighted points in $S$.
Here, $t$ is the spanning ratio of $(k, t)$-VFTS constructed in Algorithm~\ref{algo:terr1} for any set of points located on any shortest path belonging to $\cal T$.
\end{lemma}
\begin{proof}
Using induction on the number of points, we show that there exists a $3t$-spanner path between $p$ and $q$ in $G \setminus S'$ for any set $S'$ of vertices of $G$ with $|S'| \le k$.
The induction hypothesis assumes $d_{G \setminus S'}(p, q) \le 3t \cdot d_w(p, q)$ for any set $S'$ of vertices of $G$ with $|S'| \le k$ and for any two points $p, q$ in $G \setminus S'$.
Consider a set $S' \subset S$ such that $|S'| \leq k$.
Let $p, q$ be two arbitrary points in $S \setminus S'$.
For any shortest path $\pi(p, q)$ between $p$ and $q$, there exists a separator $\gamma$ such that $\pi(p, q)$ intersects $\gamma$.
Let $r$ be a point of intersection of $\pi(p, q)$ and $\gamma$.
Also, let $p_\gamma$ (resp. $q_\gamma$) be a geodesic projection of $p$ (resp. $q$) on $l$.
Since $G_\gamma$ is a $(k, t)$-VFTSWP, there exists a path $\tau$ between $p_l$ and $q_l$ in $G_l$ with length at most $t d_w(p_l, q_l)$.
By replacing each vertex $x_\gamma$ of $\tau$ by $x \in S$ such that $x_\gamma$ is the projection of $x$ on $\gamma$, gives a path between $p$ and $q$ in $G \setminus S'$. 
Thus, the length of the path $d_{G \setminus S'}(p,q)$ is less than or equal to the length of path $\tau$ in $G_\gamma$.
For every $x,y \in S \setminus S'$,
{\setlength{\abovedisplayskip}{0pt}
\begin{flalign}
\hspace{6mm}d_{G \setminus S'}(p,q) &\le \sum_{x_\gamma,y_\gamma \in \tau} d_w(x,y)&&\nonumber\\
    &= \sum_{x_\gamma,y_\gamma \in \tau} (w(x) + d_{\pi}(x,y) + w(y))&&\nonumber\\
    &\leq \sum_{x_\gamma,y_\gamma \in \tau} (w(x) + d_{\pi}(x,x_\gamma) + d_{\pi}(x_\gamma,y_\gamma) + d_{\pi}(y_\gamma,y) + w(y))&&\nonumber\\
        &\text{[by triangle inequality of geodesic shortest paths]}&&\nonumber\\
    &= \sum_{x_\gamma,y_\gamma \in \tau} (w(x_\gamma) + d_{\pi}(x_\gamma,y_\gamma) + w(y_\gamma))&&\nonumber\\
        &\text{[since the weight associated with projection} \ z_\gamma \ \text{of every point} \ z \ \text{is} \ w(z) + d_{\pi}(z,z_\gamma)]&&\nonumber\\
    	&\leq t \cdot d_w(p_\gamma,q_\gamma)&&\label{simppoly:lweq}
\end{flalign}}

{\setlength{\abovedisplayskip}{0pt}
\begin{flalign}
\hspace{1in}
        &\text{[since} \ G_\gamma \ \text{is a} \ (k, t)\text{-vertex fault-tolerant geodesic spanner]}&&\nonumber\\
        &= t \cdot [w(p_\gamma) + d(p_\gamma,q_\gamma) + w(q_\gamma)]&&\nonumber\\
        &= t \cdot [w(p_\gamma) + d_{\pi}(p_\gamma,q_\gamma) + w(q_\gamma)]&&\nonumber\\
    &\text{[since points $p_\gamma$ and $q_\gamma$ are located on $l$]}&&\nonumber\\
        &= t \cdot [w(p) + d_{\pi}(p,p_\gamma) + d_{\pi}(p_\gamma,q_\gamma) + d_{\pi}(q_\gamma,q) + w(q)]\label{eq20}&&\nonumber\\
    &\text{[since the weight associated with } \ z_\gamma \ \text{of point} \ z \ \text{is} \ w(z) + d_{\pi}(z,z_\gamma) \text{]}&&\nonumber\\
    &\leq t \cdot [w(p) + d_{\pi}(p,r) + d_{\pi}(p_\gamma,q_\gamma) + d_{\pi}(r,q) + w(q)]&&\nonumber\\
    &\text{[since $r$ is a point belonging to both $l$ as well as $\pi(p, q)$]}&&\nonumber\\
	&= t \cdot [w(p) + d_{\pi}(p, q) + w(q) + d_{\pi}(p_\gamma,q_\gamma)]&&\nonumber\\
    &= t \cdot [d_w(p, q) + d_{\pi}(p_\gamma,q_\gamma)].&&\nonumber
\end{flalign}}

\noindent
But, $d_{\pi}(p_\gamma,q_\gamma)$
{\setlength{\abovedisplayskip}{0pt}
\begin{flalign}
\hspace{0.4in}
	&\leq d_{\pi}(p_\gamma,p) + d_{\pi}(p,q) + d_{\pi}(q,q_\gamma)&&\nonumber\\
	&\text{[by triangle inequality of geodesic shortest paths]}&&\nonumber\\
        &\leq d_{\pi}(p,r) + d_{\pi}(p,q) + d_{\pi}(r,q)&&\nonumber\\
	&\text{[since $r$ belongs to both $l$ and $\pi(p,q)$]}&&\nonumber\\
        &\leq w(p) + d_{\pi}(p,r) + w(p) + d_{\pi}(p,q) + w(q) + d_{\pi}(r,q) + w(q)&&\nonumber\\
        &\text{[since weight associated with every point is non-negative]}&&\nonumber\\
    &\le 2d_w(p,q)&&\nonumber\\
    &\text{[since $r$ belongs to $\pi(p, q)$].}&\nonumber
\end{flalign}}

Hence, $d_{G \setminus S'}(p,q) \leq 3t \cdot d_w(p,q)$.
\end{proof}

From Theorem~1 in \cite{conf/caldam/Inkulu19b}, there exists a $k$-vertex fault tolerant spanner for the set $S_\gamma$ of weighted points with multiplicative stretch $4+\epsilon$. 
This immediately leads to the following theorem.

\begin{theorem}
\label{thm:terr1}
Given a set $S$ of $n$ points located on a polyhedral terrain $\cal T$, a weight function $w$ to associate a non-negative weight to each point in $S$, a positive integer $k$, and a real number $0 < \epsilon \le 1$, there exists a $(k, 12+\epsilon, w)$-vertex fault tolerant geodesic spanner with $O(\frac{k n}{\epsilon^2} \lg{n})$ edges for the weighted points in $S$.
\end{theorem}

Next, we present another algorithm (Algorithm~\ref{algo:terr2} below) to compute a spanner for a set of weighted points on a polyhedral terrain.
The spanner constructed by this algorithm has better stretch factor and its size does not depend on $\epsilon$. 
We note the number of edges in this spanner is $O(kn(\lg{n})^2)$ in contrast to $O(\frac{kn \lg{n}}{\epsilon^2})$ number of edges in the spanner constructed by Algorithm~\ref{algo:terr1}.
This algorithm mainly modifies lines $5$-$6$ of Algorithm~\ref{algo:terr1}.
First, a balanced sp-separator $\Gamma$ as given in Theorem~\ref{thm7} is computed.
For each shortest path $\gamma$ in $\Gamma$, for every point $p \in S$, we compute a geodesic projection $p_{\gamma}$ of $p$ on $\gamma$, and associate a weight $w(p) + d_{\mathcal{T}}(p, p_{\gamma})$ with $p_{\gamma}$.
Further, we compute a $\frac{8}{\epsilon}$-SSPD $\cal S$ for points projected on $\gamma$.
Using Algorithm~\ref{algo:addedgessspd}, like in algorithms for points in simple polygon and for points in polygonal domain, we introduce edges into spanner network being constructed $G$ based on $\cal S$.
The sets $S_{in}$ and $S_{out}$ comprising points are defined as follows: if the sp-separator is a shortest path $\gamma$, then $S_{in}$ is the set of points located in $\gamma^{+}(u,v) \cap S$; otherwise, $S_{in}$ is $\Delta \cap S$; and, the set $S_{out} = S - S_{in}$.
More edges are included in $G$ while recursively processing points in sets $S_{in}$ and $S_{out}$.
Refer to Algorithm~\ref{algo:terr2}.

\begin{algorithm}[ht]
    \caption{VFTSWPTerrain2(${\cal T}, S, k, \epsilon$)}
    \label{algo:terr2}

    \SetKwInOut{Input}{Input}
    \SetKwInOut{Output}{Output}
    
    \Input{A triangulated polyhedral terrain $\cal T$, a set $S$ on $n$ points located on $\cal T$, a weight function $w$ that associates a non-negative weight to each point in $S$, an integer $k \geq 1$, and a real number $0 < \epsilon \le 1$.}
    \Output{A $(k, 6+\epsilon, w)$-VFTSWP $G$.}
    
	\BlankLine

    \While{$|\mathcal{T} \cap S| \geq 1$}{
        
	\BlankLine

	Using the algorithm given in \cite{conf/soda/AbamBS17}, compute a balanced sp-separator $\Gamma$ for $\mathcal{T}$. \\
        
	\BlankLine

        \ForEach{shortest path $\gamma$ that is bounding $\Gamma$}{
			
		\BlankLine

		For each point $p$ that lies on $\cal T$, compute a geodesic projection $p_\gamma$ of $p$ on $\gamma$. 
		Let $S_\gamma$ be the set of points resultant from these projections. \\

		\BlankLine

		Using the algorithm from \cite{journals/dcg/ABFG09}, compute a $\frac{8}{\epsilon}$-SSPD $\cal S$ for points in $S_\gamma$. \\

		\BlankLine

		IncludeEdgesUsingSSPD(${\cal S}, G$).  \scriptsize{} (Refer to Algorithm~\ref{algo:addedgessspd}.) \normalsize{}

		\BlankLine

        }

	\BlankLine

        \scriptsize{}
        Let ${\cal T}'$ be $\pi^{+}(u,v)$ if the balanced sp-separator is a shortest path $\pi(u,v)$; otherwise, let ${\cal T}'$ be $\Delta$.
        Also, let $S_{in}$ be the set of points in $S$ located on the polyhedral terrain ${\cal T}'$. \\
        \normalsize{}

	\BlankLine

        VFTSWPTerrain(${\cal T}', S_{in}, k, \epsilon$). \\
    
	\BlankLine

        \scriptsize{}
        Let ${\cal T}''$ be $\pi^{-}(u,v)$ if the balanced sp-separator is a shortest path $\pi(u,v)$; otherwise, let ${\cal T}''$ be ${\cal T} \setminus \Delta$.
        Also, let $S_{out}$ be the set of points in $S$ located on the polyhedral terrain ${\cal T}''$. \\
        \normalsize{}

	\BlankLine

        VFTSWPTerrain(${\cal T}'', S_{out}, k, \epsilon$). 

	\BlankLine

    }
        
\end{algorithm}

The following lemma shows $G$ is a $(k, 6+\epsilon, w)$-VFTSWP for weighted points in $S$.

\begin{lemma}
Algorithm~\ref{algo:terr2} computes a $(k, 6+\epsilon)$-vertex fault tolerant geodesic spanner $G$ for the weighted points in $S$.
\end{lemma}
\begin{proof}
The proof is same as the proof of Lemma~\ref{lem17}.
\end{proof}

\begin{theorem}
\label{thm:terr2}
Given a set $S$ of $n$ points located on a polyhedral terrain $\cal T$, a weight function $w$ to associate a non-negative weight to each point in $S$, a positive integer $k$, and a real number $0 < \epsilon \le 1$, there exists a $(k, 6+\epsilon, w)$-vertex fault tolerant geodesic spanner with $O(k n (\lg{n})^2)$ edges for the weighted points in $S$.
\end{theorem}
\begin{proof}
The proof is same as the proof of Theorem~\ref{thm:simppoly}.
\end{proof}

\section*{Conclusions}
\label{sect:conclu}

Our first algorithm computes a $k$-vertex fault-tolerant spanner with stretch $\sqrt{10} + \epsilon$ for weighted points located in a simple polygon.
Our second algorithm computes a $k$-vertex fault-tolerant spanner with stretch $6 + \epsilon$ for weighted points located in the free space of a polygonal domain.
And, our third algorithm computes a $k$-vertex fault-tolerant spanner with stretch $6 + \epsilon$ for weighted points located on a polyhedral terrain.
It would be interesting to achieve a better bound on the stretch factor for the case of each point is unit or zero weighted.
Apart from the efficient computation, it would be interesting to explore the lower bounds on the number of edges for these problems.
Besides, the future work in the context of spanners for weighted points could include finding the relation between vertex-fault tolerance and edge-fault tolerance and optimizing various spanner parameters, like degree, diameter, and weight.

\subsection*{Acknowledgement}

This research of R. Inkulu is in part supported by SERB MATRICS grant MTR/2017/000474.

\bibliographystyle{plain}
%\bibliography{../ajar/results/bibs/geomgraphs,../ajar/results/bibs/misc,../ajar/results/bibs/shortestpaths,../ajar/results/bibs/visibility}

\end{document}